\documentclass[11pt]{article}
\usepackage{makeidx}
\usepackage{amsfonts,amsmath,graphicx,latexsym,url,epsf, algpseudocode, verbatim, float, paralist, mathtools, caption, authblk}
\usepackage{microtype}
\usepackage[margin = 1in]{geometry}
\usepackage{tikz}
\usetikzlibrary{shapes, arrows.meta}

\newcommand{\newthmwithin}[3]{\newtheorem{#1q}{#2}[#3]
                        \newenvironment{#1}{\begin{#1q}\sf}{\end{#1q}}}

\newcommand{\newthm}[3]{\newtheorem{#1q}[#2q]{#3}
                        \newenvironment{#1}{\begin{#1q}\sf}{\end{#1q}}}
																								
\newthmwithin{theorem}{Theorem}{section}
\newthm{definition}{theorem}{Definition}
\newthm{lemma}{theorem}{Lemma}
\newthm{conjecture}{theorem}{Conjecture}
\newthm{corollary}{theorem}{Corollary}

\makeatletter
\def\GrabProofArgument[#1]{ #1: \egroup\ignorespaces}
\def\proof{\noindent\textbf\bgroup Proof%
           \@ifnextchar[{\GrabProofArgument}{. \egroup\ignorespaces}}

\makeatother

\let\realbfseries=\bfseries
\def\bfseries{\realbfseries\boldmath}

\newtheorem{conj}{Conjecture}
\newtheorem{claim}{Claim}

\renewcommand{\S}{\mathbb{S}}

\newcommand{\sand}{\ensuremath{\mathsf{f}}-\ensuremath{\mathsf{SAND}}}
\newcommand{\NPhard}{\ensuremath{\mathsf{NP}}-hard}

\begin{document}
\title{Single-sink Fractionally Subadditive Network Design}

\author[1]{Guru Guruganesh \thanks{This material is based upon work supported in part by National Science Foundation awards CCF-1319811,CCF-1536002, and CCF-1617790}}
\author[1]{Jennifer Iglesias \thanks{This material is based upon work supported by the National Science Foundation Graduate Research Fellowship Program under Grant No. 2013170941} }
\author[1]{R.~Ravi \thanks{This material is based upon research supported in part by the U. S. Office of Naval Research under award number N00014-12-1-1001, and the U. S. National Science Foundation under award number CCF-1527032.}}
\author[2]{Laura Sanit\`a}
\affil[1]{Carnegie Mellon University, Pittsburgh, PA, USA \\ \texttt{\{ggurugan, jiglesia, ravi\}@andrew.cmu.edu}}
\affil[2]{University of Waterloo, Waterloo, Canada \\ \texttt{laura.sanita@uwaterloo.ca}}






\maketitle

\begin{abstract}
We study a generalization of the Steiner tree problem, where we are given a weighted network $G$ together with a collection of $k$ subsets of its vertices and a root $r$. We wish to construct a minimum cost network such that the network supports one unit of flow to the root from every node in a subset simultaneously. The network constructed does not need to support flows from all the subsets simultaneously.

We settle an open question regarding the complexity of this problem for $k=2$, and give a $\frac{3}{2}$-approximation algorithm that improves over a (trivial) known 2-approximation. Furthermore, we prove some structural results that prevent many well-known techniques from doing better than the known $O(\log n)$-approximation. Despite these obstacles, we conjecture that this problem should have an $O(1)$-approximation. We also give an approximation result for a variant of the problem where the solution is required to be a path.

\end{abstract}

\noindent \textbf{Subject Class: }F.2.2 Nonnumerical Algorithms and Problems, G2.2 Graph Theory\\
\textbf{Keywords: }Network design, single-commodity flow, approximation algorithms, Steiner tree.

\newenvironment{LabeledProof}[1]{\noindent{\bf Proof of #1: }}{\qed}

\section{Introduction}

We study a robust version of a single-sink network design problem  that we call the \emph{Single-sink fractionally-subadditive network design} (\sand{}) problem.
In an instance of \sand{}, we
are given an undirected graph $G=(V,E)$ with edge costs $w_e \geq 0$ for all
$e \in E$, a root node $r \in V$, and $k$ \emph{colors} represented as vertex
subsets $C_i \subseteq V\setminus \{r\}$ for all $i \in [k]$, that wish to send
flow to $r$.  A feasible solution is an integer capacity installation on the edges
of $G$, such that for every $i \in [k]$, each node in $C_i$ can  \emph{simultaneously} send one unit
of flow to $r$.
Thus, the total flow sent by color $i$ nodes is
is $|C_i|$ while the
flows sent from nodes of different colors are instead \emph{non-simultaneous}
and can share capacity.
An optimal solution is a feasible one that minimizes the total cost of the installation.

The single-sink nature of the problem suggests a natural
\emph{cut}-covering formulation, namely:
\begin{equation}\label{eq:ip}\tag{IP}
	\begin{split}
		\min  &\sum_{e\in E} w_e x_e \quad\text{ s.t.}\\
	 	&\sum_{e \in \delta(S)} x_e \ge f(S)\qquad\quad \forall S\subset V\setminus \{r\}\\
		&x\ge 0\,, \quad x \in \mathbb{Z} \,,
	\end{split}
\end{equation}
\noindent
where $\delta(S)$ denotes the set of edges with exactly one endpoint in $S$,
and
\begin{align}
    f(S) := \max_{i \in [k]} \{|C_i \cap S|\} \label{defn:fs}
\end{align} for all $S \subseteq V\setminus
\{r\}$. Despite having exponentially many
constraints, the LP-relaxation of (IP) can be solved in polynomial-time
because the separation problem reduces to performing $k$ max-flow computations.
The main challenge is to round the resulting solution into an integer solution.

Rounding algorithms for the LP relaxation of (IP) have been investigated by many
authors, under certain assumptions of the function $f(S)$.
Prominent examples are some classes of 0/1-functions (such as \emph{uncrossable} functions), or integer-valued functions such as \emph{proper} functions, or  \emph{weakly supermodular} functions~\cite{GW95,J01};
however, these papers consider arbitrary cut requirements rather than the single-sink connectivity requirements we study.

Our single-sink problem is a special case of a broader class of subadditive network design problems where the function $f$ is allowed to be a general subadditive function.
Despite their generality, the single-sink network design problem for general subadditive functions can be approximated within an $O(\log |V|)$ factor by using a tree drawn from the probabilistic tree decomposition of the metric induced by $G$ using the results of Fakcharoenphol, Rao, and Talwar~\cite{FRT04}, and installing the required capacity on the tree edges.
Hence, a natural direction is to consider special cases of such subadditive cut requirement functions.


Our function $f(S)$ defined in~(\ref{defn:fs}) is an interesting and important special case of subadditive functions.
It was introduced as XOS-functions (max-of-sum functions) in the
context of combinatorial auctions by Lehman et al.~\cite{LLN01}. Feige~\cite{F09} proved
that this function is equivalent to fractionally-subadditive functions
which are a strict generalization of submodular functions (hence the title).  These
functions have been extensively studied in the context of learning theory and
algorithmic game theory~\cite{BCIW12,BR11,LLN01}.
Our work is an attempt to understand their behavior as single-sink network design requirement functions.


\sand{} was first studied by Oriolo et al.~\cite{OSZ13} in the context of
\emph{robust} network design, where the goal is to install minimum cost
capacity on a network in order to satisfy a given set of (non-simultaneous)
traffic demands among terminal nodes. Each subset $C_i$ can in fact be seen as
a way to specify a distinct traffic demand that the network would like to
support.  They observed that \sand{} generalizes the Steiner tree problem:
an instance of the Steiner tree problem with $k+1$ terminals $t_1, \dots, t_{k+1}$
is equivalent to the \sand{} instance with $r:= t_{k+1}$ and
$C_i := \{ t_i \}$ for all $i \in [k]$. This immediately shows that \sand{} is \NPhard{}
(in fact, \ensuremath{\mathsf{APX}}-hard \cite{CC02}) when $k$ is part of the input. The authors in
\cite{OSZ13} strengthened the hardness result by proving that \sand{} is \NPhard{}
even if $k$ is not part of the input, and in particular for $k=3$ (if
$k=1$ the problem is trivially solvable in polynomial-time by computing a
shortest path tree rooted at $r$). From the positive side, they observed that
there is a trivial $k$-approximation algorithm, that relies on routing via
shortest paths, and an $O(\log |\cup_i C_i| )$-approximation algorithm using
metric embeddings \cite{FRT04,GNR10}. The authors conclude their paper
mentioning two open questions, namely whether the problem is polynomial-time
solvable for $k=2$, and  whether there exists an $O(1)$-approximation algorithm.

\subsection{Our results}
\begin{enumerate}
\item In this paper, we answer the first open question in~\cite{OSZ13} by showing that \sand{} is
\NPhard{} for $k=2$ via a reduction from \ensuremath{\mathsf{SAT}}.
\item We give a $\frac{3}{2}$-approximation algorithm for this case ($k=2)$. This is the
    first improvement over the (trivial) $k$-approximation obtained using shortest paths for any $k$. Our
approximation algorithm is based on pairing terminals of different groups
together, and therefore reducing to a suitable minimum cost matching problem.
While the idea behind the algorithm is natural, its analysis requires a
deeper understanding of the structure of the optimal solution.


\item We also introduce an interesting variant of \sand, which we call the
Latency-\sand{} problem, where the network built is restricted to being a
\emph{path} with the root $r$ being one of the endpoints (\sand-path). We show
a $O(\log^2 k \log n)$-approximation using a new reformulation of the problem
that allows us to exploit techniques recently developed for \emph{latency}
problems \cite{CS16}.

\item While being a generalization of well-studied problems, \sand{} does
not seem to admit an easy $O(1)$-approximation via standard
LP-rounding techniques for arbitrary values of $k$. We prove some structural results that
highlight the difficulty of the general problem (see Appendix A).  In particular, we show a family of a instances
providing a super-constant gap between an optimal \sand{} solution
and an optimal \emph{tree}-solution, i.e., a solution whose support is a tree
-- this rules out many methods that output a solution with a tree structure.
The bulk of the construction was shown in \cite{GOS11} and we amend it to our problem using a simple observation.
 Furthermore, we give some evidence that an iterative rounding approach (as in
Jain's fundamental work \cite{J01}) is unlikely to work. This follows by considering a special class of Kneser Graphs, where the LP seems to put low fractional weight on each edge in an extreme point.

\item {\bf Open Questions.}  We offer the following conjecture as our main open
    question. \footnote{Although the problem is known in some circles, it
        has not been explicitly stated as a conjecture. We do so here, in the
        hopes that it will encourage others to work on this problem.}

\begin{conj}
There exists an $O(1)$-approximation algorithm for the \sand{} problem.
\end{conj}
Although standard LP-based approaches seem to fail in providing a constant factor approximation, the worst known integrality gap example we are aware of yields a (trivial) lower bound of $2$ on the integrality
gap of (IP) for \sand. A related open question is if there is an instance of \sand{} for which the integrality gap of (IP) is greater than $2$.
\end{enumerate}

\subsection{Related work.}
Network design problems where the goal is to build a minimum cost network in
order to support a given set of flow demands, have been extensively studied in
the literature (we refer to the survey \cite{C07}). There has been a huge
amount of research focusing on the case the set of demands is described via a
polyhedron (see e.g. \cite{BK05}). In this context a very popular model is the
\emph{Virtual private network} \cite{Duff99,Fing97}, for which many
approximation results have been developed (see e.g. \cite{GOS13, GRS11,
    GKKRY01} and the references therein). For the case where the set of demands
is instead given as a (finite) discrete list, the authors in \cite{OSZ13}
developed a constant factor approximation algorithm on ring networks, and
proved that \sand{} is polynomial-time solvable on
ring networks. \\
Regarding the formulation (IP), Goemans and Williamson
\cite{GW95} gave a $O(\log (f_{max}))$-approximation algorithm for solving (IP)
whenever $f(S)$ is an integer-valued proper function that can take values up to
$f_{max}$, based on a primal-dual approach. Subsequently, Jain \cite{J01}
improved this result by giving a 2-approximation algorithm using iterative
rounding of the LP-relaxation. Recently, a strongly-polynomial time FPTAS to
solve the LP-relaxation of (IP) with proper functions has been given in
\cite{FKPS16}.

\section{$3/2$-approximation for the two color case}
The goal of this section is to give a $\frac{3}{2}$-approximation algorithm for
SAND with two colors. We remark that our algorithm bypasses the difficulties
mentioned in the previous section. In particular, the final output is not a tree.

\subsection{Simplifying Assumptions.} We will refer to the two colors
as \emph{green} and \emph{blue}, and let $C_G \subset V$ denote the set of
green terminals, and $C_B \subset V$ denote the set of blue terminals. Without
loss of generality, we will assume that $|C_G|=|C_B|$, i.e., the
cardinality of green terminals is equal to the cardinality of blue terminals
(if not,  we could easily add dummy nodes at distance 0 from the root).
Furthermore, by replacing each edge in the original graph with $|C_G|$ parallel
edges of the same cost, we can assume that in a feasible solution the capacity
installed on each edge must be either 0 or 1. This means that each edge is used
by \emph{at most} one terminal of $C_G$ (resp. $C_B$) to carry flow to the
root. Lastly, we assume that every terminal in $C_G$ shares at least one edge
with some terminal in $C_B$ in the optimal solution.\footnote{We can easily
    ensure this e.g. by modifying our instance as follows: we add a dummy node
    $r'$ which is only connected to $r$ with $|C_G|$ parallel edges of $0$ cost,
    and we make $r'$ be the new root. In this way, all terminals will use one
    copy of the edge $(r,r')$.}

Let OPT denote an optimal solution to a given instance of SAND with two
colors.  We start by developing some results on the structure of OPT, that will
be crucial to analyze our approximation algorithm later.

\subsection{Understanding the structure of OPT}
A feasible solution of a SAND instance consists of a (integer
valued) capacity installation on the edges that allows for a flow from the
terminals to the root. Given a feasible solution, each terminal will send its
unit of flow to $t$ on a single path. Let us call the collection of such paths
a \emph{routing} associated with the feasible solution.  The first important
concept we need is the concept of splits.

\subsubsection{Shared Edges and Splits.}
Given a routing, for each terminal $g \in C_G$
(and $b \in C_B$ respectively) let $P_{g}$ ($P_{b}$) denote the path along which $g$ ($b$) sends flow to the root;
i.e.~$P_{g}:=\{g=x_{0},x_{1}, \dots, x_{|P_{g}|}=r\}$.
We say that an edge $e$ is \textbf{shared}
if the paths of two terminals of different color contain the edge.
We say that $g \in C_G$ and $b \in C_B$ are \textbf{partners} with respect
to a shared edge $e=uv$, if their respective paths use the edge
$e$; i.e.~$e \in P_{g} \cap P_{b}$.

\begin{definition}
A \textbf{split} in the path $P_{g}$ is a
maximal set of consecutive edges of the path
such that $g$ is partnered with some $b$ on all the
edges of this set.
\end{definition}

If $\{(x_{i},x_{i+1}),(x_{i+1},x_{i+2}), \dots,(x_{i+j-1},x_{i+j})\}$
is a split in the path $P_{g}=\{g=x_{0},x_{1}, \dots, x_{|P_{g}|}=r\}$ for $g \in C_G$, then  there exists
a unique terminal $b \in C_B$ such that $P_b$ contains the edges
$\{(x_{i},x_{i+1}),(x_{i+1},x_{i+2}), \dots,(x_{i+j-1},x_{i+j})\}$, $P_b$ does
not contain the edge $(x_{i-1},x_{i})$, and if $x_{i+j} \neq r$ then $P_b$ does
not contain the edge $(x_{i+j},x_{i+j+1})$. By our assumptions,
the terminal $b$ is unique as each edge is used by at most one terminal of
each color.

Since the flow is going from a terminal $g$ to $r$, the path $P_g$
naturally induces an orientation on its edges given by the direction of the flow, even though the edges are
undirected. Of course, the paths of different terminals could potentially induce opposite orientations
on (some of) the shared edges (see Figure~\ref{fig1}).

\begin{definition}
    A split is \textbf{wide}, if the paths of the
    two terminals that are partners on the edges of the split induce opposite orientations on the edges.
    A split is \textbf{thin}, if the paths of the
    two terminals that are partners on the edges of the split induce the same orientation on the edges.
\end{definition}

The above notions are well defined for any routing with respect to a feasible
solution. Now, we focus on the structure of an optimal routing, i.e., a
routing with respect to an optimal solution.  For the rest of this
section, we let $\{P_{g}\}_{g \in C_G}$ and $\{P_{b}\}_{b \in C_B}$ be an
optimal routing. The following lemma is immediate.

\begin{lemma} \label{lem:sp}
    Let $\{(x_{i},x_{i+1}),(x_{i+1},x_{i+2}), \dots,(x_{i+j-1},x_{i+j})\}$ be a
    split in the path $P_{g}$ (for some $g \in C_G$). The edges of the split
    form a shortest path from $x_{i}$ to $x_{i+j}$.
\end{lemma}
\begin{proof}
    If not, we could replace this set of edges with the set of edges of a
    shortest path  from $x_{i}$ to $x_{i+j}$, in both $P_g$ and $P_b$, where
    $b$ is the partner of $g$ on the split. Therefore, we can install one unit
    of capacity on these edges, and remove the unit of capacity from the edges
    of the split. We get another feasible solution with smaller cost, a
    contradiction to the optimality of our initial solution.
\end{proof}

\subsubsection{Split Graph.}
A consequence of Lemma \ref{lem:sp} is each split is entirely
characterized by the endpoints of the split and the terminals that share
them. We denote each split by a tuple $(u,v,g,b)$ which states
that there is a shortest path between $u$ and $v$ whose edges are shared by $g$
and $b$.

\begin{figure}
\centering
	\begin{tikzpicture}
	\node[circle, fill=black,thick, inner sep=2pt, minimum size=0.1cm, label=above:$r$](r) at (0,5) {};
	\node[circle, fill=black,thick, inner sep=2pt, minimum size=0.1cm](1) at (-1.5,4) {};
	\node[circle, fill=black,thick, inner sep=2pt, minimum size=0.1cm](2) at (1.5,4) {};
	\node[circle, fill=black,thick, inner sep=2pt, minimum size=0.1cm](3) at (-2,2.5) {};
	\node[circle, fill=black,thick, inner sep=2pt, minimum size=0.1cm](4) at (2,2.5) {};
	\node[circle, fill=black,thick, inner sep=2pt, minimum size=0.1cm](5) at (-1,1) {};
	\node[circle, fill=black,thick, inner sep=2pt, minimum size=0.1cm](6) at (1,1) {};
	\node[rectangle, fill=blue,thick, inner sep=2pt, minimum size=0.2cm, label=below:$b_2$](b2) at (-3,1) {};
	\node[rectangle, fill=blue,thick, inner sep=2pt, minimum size=0.2cm, label=below:$b_1$](b1) at (-1.5,0) {};
	\node[regular polygon, regular polygon sides=3, fill=green,thick, inner sep=2pt, minimum size=0.1cm, label=below:$g_2$](g2) at (3,1) {};
	\node[regular polygon, regular polygon sides=3, fill=green,thick, inner sep=2pt, minimum size=0.1cm, label=below:$g_1$](g1) at (1.5,0) {};
	\draw[->, -{Latex[length=3mm]}, blue] (b1)-- (5);
	\draw[->, -{Latex[length=3mm]}, green] (5)-- (3);
	\draw[->, -{Latex[length=3mm]}, cyan] (3)-- (1);
	\draw[->, -{Latex[length=3mm]}, cyan] (1)-- (r);
	\draw[->, -{Latex[length=3mm]}, green] (g1)-- (6);
	\draw[->, -{Latex[length=3mm]}, blue] (6)-- (4);
	\draw[->, -{Latex[length=3mm]}, cyan] (4)-- (2);
	\draw[->, -{Latex[length=3mm]}, cyan] (2)-- (r);
	\draw[->, -{Latex[length=3mm]}, blue] (b2)-- (3);
	\draw[->, -{Latex[length=3mm]}, green] (g2)-- (4);
	\draw[->,-{Latex[length=3mm]}, blue] (5) to[out=30,in=150] (6);
	\draw[->,-{Latex[length=3mm]}, green] (6) to[out=-150, in=-30] (5);
\node[circle, fill=black,thick, inner sep=2pt, minimum size=0.1cm, label=above:$g_1-b_2$](11) at (7,4) {};
	\node[circle, fill=black,thick, inner sep=2pt, minimum size=0.1cm, label=above:$g_2-b_1$](12) at (9,4) {};
	\node[circle, fill=black,thick, inner sep=2pt, minimum size=0.1cm, label=right:$g_1-b_1$](13) at (8,2) {};
	\node[rectangle, fill=blue,thick, inner sep=2pt, minimum size=0.2cm, label=below:$b_2$](1b2) at (6,2) {};
	\node[rectangle, fill=blue,thick, inner sep=2pt, minimum size=0.2cm, label=below:$b_1$](1b1) at (7,0) {};
	\node[regular polygon, regular polygon sides=3, fill=green,thick, inner sep=2pt, minimum size=0.1cm, label=below:$g_2$](1g2) at (10,2) {};
	\node[regular polygon, regular polygon sides=3, fill=green,thick, inner sep=2pt, minimum size=0.1cm, label=below:$g_1$](1g1) at (9,0) {};
	\draw[->, -{Latex[length=3mm]}, blue] (1b1)-- (13);
	\draw[->, -{Latex[length=3mm]}, green] (13)-- (11);
	\draw[->, -{Latex[length=3mm]}, blue] (13)-- (12);
	\draw[->, -{Latex[length=3mm]}, blue] (1b2)-- (11);
	\draw[->, -{Latex[length=3mm]}, green] (1g2)-- (12);
	\draw[->, -{Latex[length=3mm]}, green] (1g1)-- (13);

	\end{tikzpicture}

    \caption{The above left graph (where each undirected edge is supposed to have unit capacity) shows an optimal routing for some \sand{} instance. Note that $b_1$ and $g_2$ (resp. $b_2$ and $g_1$) send flow to $r$ going counterclockwise (resp. clockwise) on the edges of the cycle. The path $P_{b_1}$ contains two splits: the first is wide ($b_1$ is partnered with $g_1$), the second is thin ($b_1$ is partnered with $g_2$).\\ The graph on the right is the Split Graph for the optimal solution on the left. The pair of vertices $g_1,b_1$ and the pair of vertices $g_2,b_2$ constitute the fresh pairs.}
		\label{fig1}
\end{figure}
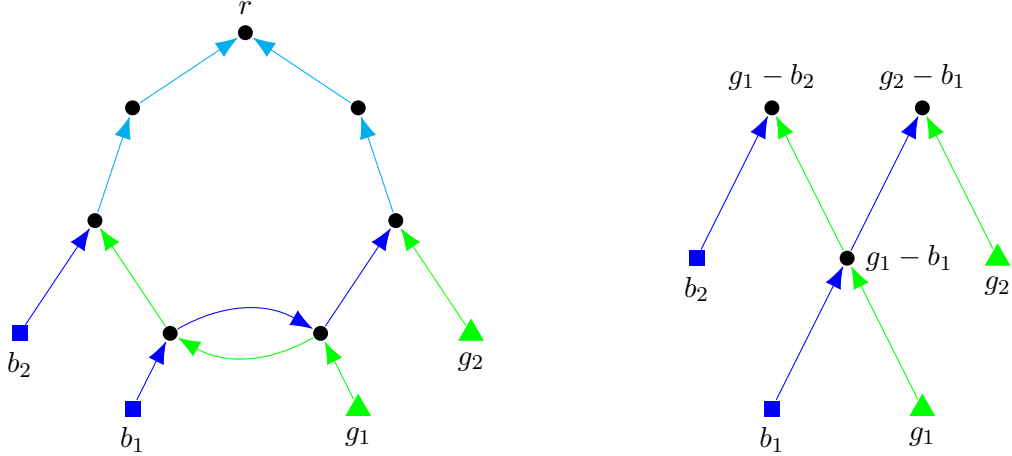

Let $\S$ denote the set of all splits in the optimal routing. We construct a
directed graph $G^{\S}$ whose vertex set corresponds to $V = \S \cup C_G \cup
C_B$ (i.e. the vertex set contains one vertex for each split and one vertex for
each terminal). For each $g \in C_{G}$, we place a directed \emph{green} edge
going between two consecutive splits in $P_{g}$.  Specifically, if
$\{(x_{i},x_{i+1}), \dots,(x_{i+j-1},x_{i+j})\}$ and $\{(x_{i'},x_{i'+1}),
\dots,(x_{i'+j'-1},x_{i'+j'})\}$ are two splits in $P_{g}$ with $i < i'$, we
say that they are consecutive if the subpath from $x_{i+j}$ to $x_{i'}$ does
not contain any split. In this case, we place a directed edge  in $G^{\S}$
whose tail is the vertex corresponding to the first split, and whose head is
the vertex corresponding to the second one.  Similarly, for each $b \in C_{B}$
we place a directed \emph{blue} edge between vertices of consecutive splits
that appear in $P_{b}$.  Furthermore, for each $g \in C_G$ (resp. $b \in C_B$)
we place a directed green (resp. blue) edge from $g$ (resp. $b$) to the vertex
corresponding to the first split on the path $P_g$ (resp. $P_b$), if any.
This graph is denoted as the \textbf{Split Graph} (see Figure~\ref{fig1}).

Each split indicates that two terminals of different colors are
sharing the capacity on a set of edges in an optimal routing. Hence, each
split-vertex in $G^{\S}$ has indegree 2 (in particular, one edge of each
color).  Furthermore, each split-vertex in $G^{\S}$ has outdegree either $0$ or
$2$; if it has two outgoing edges, one is green and one is blue. Similarly,
each terminal has indegree $0$, and outdegree $1$ (as we assume that
each terminal shares at least one edge).

\subsubsection{Fresh Pairs.}

We need one additional definition before proceeding to the algorithm.
\begin{definition}
     An \textbf{$\S$-alternating sequence} is a sequence of vertices of the Split Graph $\{v,s_1,s_2,\dots,s_h,w\}$ with $h\geq 1$,
    that satisfies the following:
    \begin{enumerate}
        \item[(i)] $(v,s_1)$ and $(w,s_h)$ are directed edges in $G^{\S}$
              and $v,w$ are terminals of different color.
        \item[(ii)] For all even $i \geq 2$, $(s_i,s_{i-1})$ and $(s_i,s_{i+1})$ are
              both directed edges in $G^{\S}$ with opposite colors.
    \end{enumerate}
    We call the path obtained by taking the edges in (i) and (ii) an \textbf{$\S$-alternating path}.
		We call $(v,w)$  a \textbf{fresh pair} if they are the endpoints of an $\S$-alternating path.
\end{definition}
By definition, in an $\S$-alternating sequence the vertices $s_1, \dots, s_h$ are all split-vertices,
and $h$ is odd. We remark here that an $\S$-alternating path is \emph{not} a directed path.
(See again Figure~\ref{fig1}).

\begin{lemma}
    \label{lem:splitpartition}
    We can find a set of edge-disjoint $\S$-alternating paths in the Split
    Graph such that each terminal is the endpoint of exactly one path in this
    set.
\end{lemma}
\begin{proof}
    We construct the desired set as follows.
    For each vertex $g \in C_G$, there is a unique outgoing
    edge to a split vertex $s \in \S$ (as we assume every terminal participates
    in a split). Since each split-vertex has indegree $2$, $s$ has another
    ingoing edge coming from a different vertex $w$.  If $w \in
    C_B$, then $(v,w) $ is a fresh pair and we have found an $\S$-alternating
    sequence $\{v,s,w\}$.  If $w$ is a
    split-vertex, then it has another outgoing edge to a different split-vertex
    $s'$, which in its turn has another incoming edge from a different vertex $w'$.
    We continue to build an alternating sequence (and a corresponding alternating
    path) in this way until it terminates in a terminal. Since the path is of even
    length and the colors alternate, we can conclude that this will terminate in a
    terminal of opposite color. We remove the edges of this path from the Split
    Graph, and iterate the process. Each terminal will belong to
    exactly one $\S$-alternating path, as it has outdegree exactly 1, and all
    the paths are edge-disjoint, proving the lemma.
\end{proof}

\subsection{The Algorithm}
We are now ready to present our \emph{matching algorithm}. The algorithm has two steps.
First, construct a complete bipartite
graph $\mathcal H$ with the bipartitions $C_G$ and $C_B$, where the weight on
the edge $(g,b) \in C_G \times C_B$ is equal to the cost of the Steiner tree in
$G$ connecting $g, b$ and the root.  Note that the graph $\mathcal H$ can be computed in
polynomial time, since a Steiner tree on $3$ vertices can be easily computed in
polynomial time.

Second, find a minimum-weight perfect matching $\mathcal{M}$ in $\mathcal H$, and for
each edge $(g,b) \in \mathcal M$  install (cumulatively) one unit of capacity
on each edge of $G$ that is in the Steiner tree associated to the edge $(g,b)
\in \mathcal M$.  The capacity installation output by
this procedure is a feasible solution to \sand{}, and has total cost equal
to the weight of $\mathcal M$.

\begin{lemma}
    The matching algorithm is a $\frac32$-approximation algorithm.
\end{lemma}

\begin{proof}
First, we partition $OPT$ into four parts; let $w_b$ (and $w_g$ respectively) be the cost of the
edges which are used only by blue (green respectively) terminals in $OPT$, and let $w_t$
($w_d$) be the cost of edges in thin (wide) splits in $OPT$. Thus,
$w(OPT) = w_b + w_g + w_t + w_d$.
By Lemma~\ref{lem:splitpartition}, we can extract from the Split Graph associated to $OPT$ a set of
$\S$-alternating paths such that each terminal is contained in exactly one fresh pair.
Consider the matching $\mathcal M_1$ determined by the set of fresh pairs found by the
aforementioned procedure.
We will now bound the weight of $\mathcal M_1$.
\begin{claim}
    The weight of the matching formed by connecting the fresh pairs
    is at most
    $$\frac32 \cdot w_b + \frac32\cdot w_g + 1\cdot w_t+ 3\cdot w_d.$$
\end{claim}

\begin{proof}
Let $(g,b)$ be a fresh pair and $(g,s_1,\dots,s_h,b)$ be
the corresponding $\S$-alternating sequence. The edges of the associated
$\S$-alternating path naturally correspond to paths in $G$ composed by
non-shared edges (that connect either the endpoints of two different splits, or
one terminal and one endpoint of a split).  These paths together with the edges
of the wide splits in the sequence, naturally yield a path $P(b,g)$ in $G$
connecting $g$ and $b$.

If we do this for all fresh pairs, we obtain that the total cost of the paths
$P(b,g)$ is upper bounded by $1\cdot w_b + 1 \cdot w_g + 2 \cdot w_d$. The
reason for having a coefficient of 2 in front of $w_d$ is because the
$\S$-alternating paths of Lemma~\ref{lem:splitpartition} are edge-disjoint, but
not necessarily vertex-disjoint: however, since each split-vertex has at most 4
edges incident into it, it can be part of at most 2 $\S$-alternating paths.

Using the aforementioned connection, we can move all terminals in $C_G$ to their
partners in $C_B$. Subsequently, we connect them to the root using the $P_b$ for all $b \in C_B$.
This connection to the root will incur a cost of $1 \cdot w_b + 1\cdot w_d + 1\cdot w_t$.
Combining this together, we get a total cost of $2\cdot w_b +1\cdot w_g + 1\cdot w_t+ 3\cdot w_d$.
Analogously, if we connect the partners in $C_G$ to the root using the the path
$P_g$ for all $g \in C_G$, we will incur a total cost of
$1\cdot w_b +2\cdot w_g + 1\cdot w_t+ 3\cdot w_d$.
Since the sum of the cost of the Steiner trees connecting the fresh pairs to the root is no more
than either of these two values, we can bound the weight of $\mathcal M_1$ by their average:
$$\frac32 \cdot w_b + \frac32\cdot w_g + 1\cdot w_t+ 3\cdot w_d. $$
\end{proof}

\begin{claim}
    There exists a matching in $\mathcal{H}$ of weight at most $1\cdot w_b +1\cdot w_g + 2\cdot w_t$.
\end{claim}
\begin{proof}
    Consider the flow routed on the optimal paths by the set of all terminals $C_G \cup C_B$.
    We modify the flow (and the corresponding routing) as follows. Whenever two terminals traverse a wide-split, re-route the flows so
    as to not use the wide-split. This is always possible as the two terminals traverse
    these edges in opposite directions (by definition of wide splits). This re-routing
    ensures that all the edges of wide-splits are not used anymore in the resulting paths. However, thin-splits
    which contained terminals of different colors passing in the same direction, might now contain
    two terminals of the same color passing through the edges. This means that these edges
    will be used twice (or must have twice the capacity installed). All other edges do not need to
    have their capacity changed.  Thus, the resulting
    flow can be associated with a feasible solution of cost at most $1\cdot w_b +1\cdot w_g + 2\cdot w_t + 0\cdot w_d$. This flow
    corresponds to all vertices directly connecting to the root as any shared edge is counted twice.
    Hence, this is a bound on any matching in $\mathcal{H}$.
\end{proof}

\noindent The average weight of the above matchings is an upper-bound on the minimum weight of a matching in $\mathcal{H}$. Hence, the weight of $\mathcal M$ is at most
		\begin{align*}
            & \frac{1}{2}\cdot \big( \frac32\cdot w_b + \frac32\cdot w_g + 1\cdot w_t + 3\cdot w_d \big) +
            \frac{1}{2}\cdot \big( 1\cdot w_b +1\cdot w_g + 2\cdot w_t + 0\cdot w_d \big)  \\
            &\leq\frac{3}{2}\cdot \big( w_b + w_g +  w_t +  w_d)
		\end{align*}
    Therefore, the matching algorithm is $\frac{3}{2}$-approximation algorithm.
\end{proof}

\section{Hardness for two colors}

We prove that the SAND problem is NP-hard even with just two colors.
\begin{theorem}
The SAND problem with $2$ colors is NP-hard.
\end{theorem}
\begin{proof}
We use a reduction from a variant of the Satisfiability (SAT) problem, where
each variable can appear in at most 3 clauses, that is known to be NP-hard
\cite{Y78}.  Formally, in a SAT instance we are given $m$ clauses $K_1, \dots,
K_m$, and $p$ variables $x_1, \dots, x_p$.  Each clause $K_j$ is a disjunction
of some \emph{literals}, where a literal is either a variable $x_i$ or its
negation $\bar x_i$, for some $i$ in $1, \dots, p$.  The goal is to find a
truth assignment for the variables that satisfies all clauses, where a clause
is satisfied if at least one of its literals takes value \emph{true}. In the
instances under consideration, each variable $x_i$ appears in at most 3
clauses, either as a literal $x_i$, or as a literal $\bar x_i$.  It is not
difficult to see that, without loss of generality, we can assume that every
variable appears in exactly 3 clauses. Furthermore, by possibly replacing all
occurrences of $x_j$ with $\bar x_j$ and vice versa, we can assume that each
variable $x_i$ appears in exactly one clause in its negated form ($\bar x_i$).

Given such a SAT instance, we define an instance of SAND as follows (see Fig. \ref{fig:reduction}). We construct a graph $G=(V,E)$ by introducing one sink node $r$, one node $k_j$ for each clause $K_j$, and 7 distinct nodes $y^\ell_i$, ($\ell=1,\dots, 7)$, for each variable $x_i$.
That is,
$$V:=  \{r\} \cup \{ k_1, \dots, k_m\} \cup \left\{\bigcup_{i=1}^p \{y^1_i, y^2_i, y^3_i, y^4_i, y^5_i, y^6_i, y^7_i\}\right\}   $$

The set of edges $E$ is the disjoint union of three different sets, $E:=E_1 \cup E_2 \cup E_3$, where:
$$E_1 := \bigcup_{i=1}^p \left\{  \bigcup_{\ell =1}^{4} \{r,y_i^{2\ell -1}\} \right\}; \quad E_2 :=  \bigcup_{i=1}^p \left\{ \bigcup_{\ell =1}^{6} \{y_i^\ell, y_i^{\ell +1} \} \right\}.$$

To define the set $E_3$, we need to introduce some more notation. For a variable $x_i$, we let $i_1$ and $i_2$ be the two indices of the clauses containing the literal $x_i$, and we let $i_3$ be the index of the clause containing the literal $\bar x_i$. We then have
$$E_3 := \bigcup_{i=1}^p   \left\{ \{y_i^2, k_{i_1}\}, \{y_i^4, k_{i_3}\}, \{y_i^6, k_{i_2}\}   \right\}.$$
We assign cost 2 to the edges in $E_1$, unit cost to the edges in $E_2$, and a big cost $M >> 0$ to the edges of $E_3$ (in particular, $M > 2m + 8p$).
Finally, we let the color classes\footnote{We here have $C_1 \cap C_2 \neq \emptyset$. However, the reduction can be easily modified to prove hardness of instances where $C_1 \cap C_2 = \emptyset$, by simply adding for all $j$ two nodes $k_j^1, k_j^2$ adjacent to $k_j$ with an edge of zero cost, and by letting $k_j^1$ (resp. $k_j^2$) be in $C_1$ (resp. $C_2$) instead of $k_j$.} be defined as:
$$C_1 := \{k_1, \dots, k_m\} \cup \left\{ \bigcup_{i=1}^p \{y_i^1, y_i^5\}\right\}; \quad C_2 :=  \{k_1, \dots, k_m\} \cup \left\{ \bigcup_{i=1}^p \{y_i^3, y_i^7\}\right\}.$$
\vspace*{-.5cm}
\begin{figure}
\centering
 \includegraphics[height=4cm]{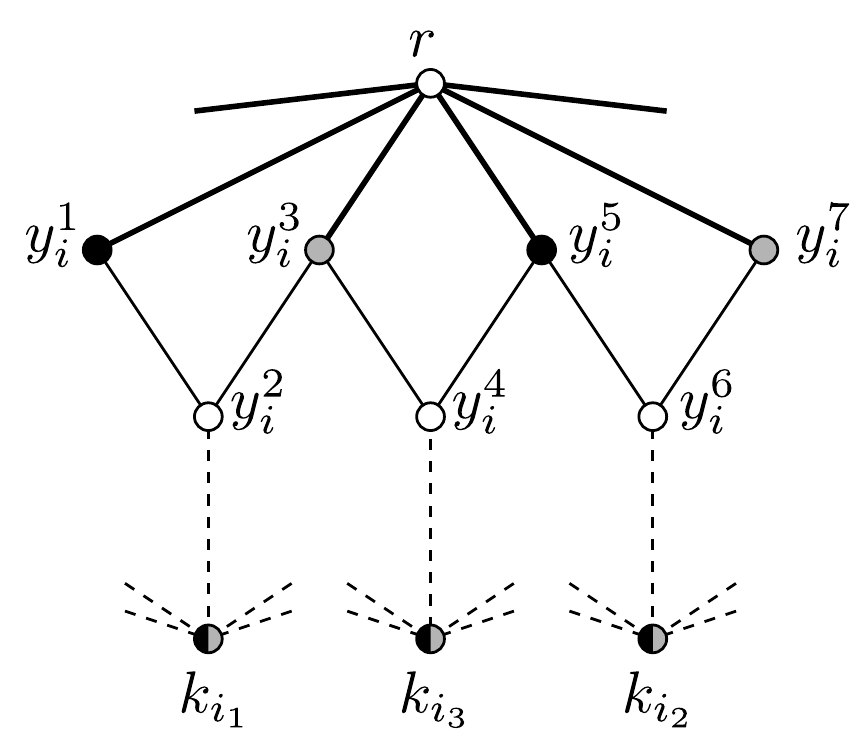}
 \caption{The picture shows the subgraph introduced for every variable $x_i$. Bold edges have cost 2, solid edges have cost 1, and dashed edges have cost $M$. Black circles indicate nodes in $C_1$, and grey circles indicate nodes in $C_2$. Nodes in $C_1 \cap C_2$ are colored half-black and half-grey.}
 \label{fig:reduction}
\end{figure}

We claim that there exists an optimal solution to the SAND instance of cost at most $(M+2)m + 8p$ if and only if there is a truth assignment satisfying all clauses for the SAT instance.

\subsection{Completeness}
 First, let us assume that the SAT instance is satisfiable. For each clause $K_j$, we select one literal that is set to \emph{true} in the truth assignment.
We define the paths for our terminal nodes in $C_1$ as follows.
For each node $y \in \bigcup_{i=1}^p \{y_i^1, y_i^5\}$, we let the flow travel from $y$ to $r$ along the edge $\{y,r\}$.
For each $k_j$, we let the flow travel to $r$ on a path $P_1^j$, that we define based on the literal selected for $K_j$.
Specifically, let $x_i$ be the variable corresponding to the literal selected for the clause $K_j$. Then:
\begin{itemize}
\item if $K_j =K_{i_1}$, we let $P_1^j$ be the path with nodes $\{k_j, y_i^2, y_i^3, r\}$,
\item if $K_j =K_{i_2}$, we let $P_1^j$ be the path with nodes $\{k_j, y_i^6, y_i^7, r\}$,
\item if $K_j =K_{i_3}$, we let $P_1^j$ be the path with nodes $\{k_j, y_i^4, y_i^3, r\}$.
\end{itemize}
We define the paths for our terminal nodes in $C_2$ similarly.
For each node $y \in \bigcup_{i=1}^p \{y_i^3, y_i^7\}$, we let the flow travel from $y$ to $r$ along the edge $\{y,r\}$. For each $k_j$,  we let the flow travel to $r$ on a path $P_2^j$ defined as follows. Let $x_i$ be the variable corresponding to the literal selected for the clause $K_j$. Then:
\begin{itemize}
\item if $K_j =K_{i_1}$, we let $P_2^j$ be the path with nodes $\{k_j, y_i^2, y_i^1, r\}$,
\item if $K_j =K_{i_2}$, we let $P_2^j$ be the path with nodes $\{k_j, y_i^6, y_i^5, r\}$,
\item if $K_j =K_{i_3}$, we let $P_2^j$ be the path with nodes $\{k_j, y_i^4, y_i^5, r\}$.
\end{itemize}

Note that the paths of terminals belonging to the same color set do not share edges.
In fact, by construction, the paths of two terminals in $C_1$ could possibly share an edge only if
for two distinct clauses $K_j \neq K_{j'}$ we selected a literal corresponding to the same variable $x_i$, and we have
$K_j =K_{i_1}$ and $K_{j'} =K_{i_3}$, since in this case the paths $P_1^j$ and $P_1^{j'}$ would share the edge $\{y_i^3,r\}$. However, selecting $x_i$ for $K_{i_1}$ means $x_i$ takes value \emph{true} in the truth assignment, while
selecting $x_i$ for $K_{i_3}$ means $x_i$ takes value \emph{false} in the truth assignment, which is clearly a contradiction.
A similar observation applies to paths of terminals in $C_2$. It follows that installing one unit of capacity on every edge that appears in (at least) one selected path is enough to support the flow of both color sets.
The total installation cost is exactly $8p + (M+2)m$.

\subsection{Soundness} Suppose there is an optimal solution to the SAND instance of cost at most $(M+2)m + 8p$. Let $\mathcal S$ denote such solution. Since the support of any feasible solution has to include at least one distinct edge of cost $M$ for each node $k_j$, and $M > 2m + 8p$, it follows that $\mathcal S$ has exactly $m$ edges of cost $M$ in its support, each with one unit of capacity installed. Hence, if we denote by $P_1^j$ (resp. $P_2^j$) the path used by $k_j$ to send flow to $r$ with terminals in $C_1$ (resp. $C_2$), we have the following fact.

\smallskip
\noindent
\emph{Fact 1}.
For each $j=1, \dots, m$, the paths $P_1^j$ and $P_2^j$ from $k_j$ to $r$ share the first edge.

\smallskip
We use this insight to construct a truth assignment for the SAT variables. Specifically,
let $y_i^{\ell}$ be the endpoint of the first edge of $P_1^j$ and $P_2^j$.
We set $x_i$ to \emph{true} if $y_i^{\ell} = y_i^{2}$ or if  $y_i^{\ell} = y_i^{6}$, and we set $x_i$ to \emph{false} if $y_i^{\ell} = y_i^{4}$.
We repeat this for all clauses $j=1, \dots, m$, and we assign an arbitrary truth value to all remaining variables, if any.
In order to finish the proof, we have to show that this assignment is consistent for all $i=1, \dots, p$. To this aim,
let us say that a variable $x_i$ is \emph{in conflict} if there is a node $k_j$ sending flow to $r$ on a path whose first edge has endpoint $y_i^4$, and there is node $k_{j'} \neq k_j$ sending flow to $r$ on a path whose first edge has endpoint $y_i^2$ or $y_i^6$. Note that our assignment procedure is consistent and yields indeed a valid truth assignment if and only if there is no variable in conflict.

We now make a few claims on the structure of $\mathcal S$, that will be useful to show that no variable can be in conflict. Next fact follows from basic flow theory.

\smallskip
\noindent
\emph{Fact 2.} Without loss of generality, we can assume that the flow sent from terminals in $C_1$ (resp. $C_2$) to $r$, does not induce directed cycles.

\smallskip
\begin{claim}
    \label{clm:nodouble}
Without loss of generality, we can assume that every terminal
sends flow to $r$ on a path that contains exactly one node $y \in \bigcup_{i=1}^p\{y_i^1,y_i^3,y_i^5,y_i^7\}$.
\end{claim}
We defer the proof of this claim which is central to the remaining proof to the end.
Let $G_i$ be the subgraph of $G$ induced by the nodes $\{r, y_i^1, \dots, y_i^7\}$,
and let $\chi_i$ be the total cost of the capacity that $\mathcal S$ installs on the subgraph $G_i$.  Note that, by Fact 1, the cost of $\mathcal S$ is $m\cdot M + \sum_{i=1}^m \chi_i$.
We will use Claim 1 to give a bound on the value $\chi_i$. To this aim, let $n_i$ be the number of nodes $k_j$ whose path $P_1^j$ contains edges of $G_i$. Note that $0 \leq n_i \leq 3$, and each $k_j$ contributes to exactly one $n_i$, for some $i=1, \dots, p$.

\begin{claim}
    \label{clm:second}
    We have $\chi_i \geq 8 + 2n_i$, with the inequality being strict if the variable $x_i$ is in conflict.
\end{claim}

Claim~\ref{clm:second}  finishes our proof, since it implies that the cost of $\mathcal S$ is at least
\begin{align*}
    m\cdot M + \sum_{i=1}^p \chi_i \geq m\cdot M + \sum_{i=1}^p (8+2n_i) = m\cdot M +8p+2m,
\end{align*}
with the inequality being tight if and only if there is no variable in conflict.
\end{proof}

\section{Latency SAND}

As described in Section~\ref{sec:obstacles}, there is a $\Omega(\log n)$ gap
between the tree and graph version of \sand{}. This naturally raises the
question of approximating \sand{} when the solution must be restricted to
different topologies. In this section, we consider the \sand{} when the output
topology must be a path. Since this variant of \sand is not easy to solve on a tree,
it is not clear how to solve it using tree metrics.

\begin{definition}
    In the \emph{latency-\sand} problem, we are given an instance of \sand, but
    require the output to be a path with the
    root $r$ as one of its endpoints. Our goal is output a minimum cost
    path, where the cost of an edge is $w_e \cdot$ (load on $e$). The load on an edge is the maximum number of nodes of one color it separates from the root.
\end{definition}

We assume that the lengths are integers and polynomially bounded in the input and give a
time-indexed length formulation for this problem. This linear programming
formulation was introduced by Chakrabarty and Swamy~\cite{CS16} for
orienteering problems.

\textbf{The Linear Programming Formulation for Latency-\sand}
\begin{alignat}{3}
    \min & \quad & \sum_{j,t} t\cdot x_{j,t} & \tag{LP$_{\mathcal{P}}^{b}$} \label{lpb} \\
    \text{s.t.} && \sum_t x_{j,t} & \geq 1 \qquad && \forall j  \in [m]  \label{jasgn} \\
    && \sum_{P\in\mathcal{P}_{b \cdot t}} z_{P,t} & \leq 1 \qquad && \forall t \in [T] \label{onep} \\
    && \sum_{P\in\mathcal{P}_{b \cdot t}: j\in P} z_{P,t} & \geq \sum_{t'\leq t}x_{j,t'} \qquad && \forall j \in [m],t \in [T] \label{jcov} \\
    && x, z & \geq 0 \notag
\end{alignat}

We assume without loss of generality, that $|C_i| = m$ for all $i\in [k]$.
$\mathcal{P}_t$ denotes the set of paths of weight at most $t$ starting from
the root. Since the lengths are polynomially bounded, we can contain a variable for each
possible length (we denote $T$ to be the maximum possible length). We use $j \in P_t$ to indicate that the path $P_t$ contains $j$
terminals of each color. The variable $x_{j,t}$ indicates that we have seen $j$
terminals of each color by time $t$ and $z_{P,t}$ indicates that we use path
$P$ to visit the terminals at time $t$.

\begin{lemma}
The linear program~\ref{lpb} is a relaxation of Latency-\sand for $b \geq 1$.
\end{lemma}
\begin{proof}
    We show that the contraints and objective are valid for any feasible solution to Latency-\sand.
\begin{itemize}
    \item Constraint~\ref{jasgn} ensures that $j$ terminals of each color are covered at some given time period, for every $j \in [m]$.
    \item Constraint~\ref{onep} ensures that only one path is (fractionally) picked for each time period $t$.
    \item Constraint~\ref{jcov} indicates that we must have picked a path $P$ that covers $j$ terminals by time $t$ if $\sum_{t' \leq t} x_{j,t'} =1$.
    \item The objective function correctly captures the cost of the path.  For
        an integer solution, $x_{j,t}=1$ indicates that time $t$ is the first
        time $j$ terminals of each color are present in the path.  Thus the
        objective counts the prefix length $t^1$ corresponding to where
        $x_{1,t^1}=1$ in all $m$ of the terms, the next prefix of length $t^2 -
        t^1$ in $m-1$ of them and so on. This accurately accounts for the loads
        in these segments of the path according to the objective function in
        \sand{}. Finally, $b \geq 1$ only allows the paths to be of lengths
        longer by a factor of $b$ so keeps the optimal solution feasible.
\end{itemize}
\end{proof}

First, we can relax the above LP by replacing $\mathcal{P}_t$ with
$\mathcal{T}_t$ which is the set of all trees of size at most $t$. This is a
relaxation as $\mathcal{P}_t\subseteq \mathcal{T}_t$.
Lemma~\ref{lem:round}, shows that we can round LP$_{\mathcal{T}}^b$ to get a
$O(b)$ approximation to latency-\sand.

\begin{lemma}
\label{lem:round}
    Given a fractional solution $(x,z)$ to LP$_{\mathcal{T}}^b$, we can round it to a
    solution to latency-\sand{} with cost at most $O(b)$ times the cost of LP$_{\mathcal{T}}^b$.
\end{lemma}
We defer the proof to the appendix due to space constraints but briefly sketch the argument.
Roughly, we sample the trees at geometric intervals and ``eulerify'' them to produce a solution
whose cost is not too much larger than the LP-objective.

Despite, being able to round the LP, we cannot hope to solve it effeciently due
to the exponential number of variables in the primal.  We will use the dual to
obtain a solution to a relaxed version of the primal.
\begin{alignat}{3}
    \max & \quad & \sum_j\alpha_j & - \sum_t\beta_t \tag{Dual$_\mathcal{P}^{b}$} \label{dualb} \\
    \text{s.t.} && \alpha_j & \leq t+\sum_{t'\geq t}\theta_{j,t'} \qquad && \forall j,t \label{d1} \\
    && \sum_{j\in P}\theta_{j,t} & \leq \beta_t \qquad && \forall t, P\in \mathcal{P}_{bt } \label{dineq1} \\
    && \alpha, \beta, \theta & \geq 0. \label{d3}
\end{alignat}

Following~\cite{CS16} it is suffient that an ``approximate separation oracle'' in the sense of
Lemma~\ref{lem:apxseparate} is sufficient to compute an optimal solution to LP$_\mathcal{T}^b$.

\begin{lemma}
\label{lem:apxseparate}
    Given a solution $(\alpha,\beta,\theta)$, we can show that either $(\alpha,
    \beta, \theta)$ is a solution to Dual$_\mathcal{T}^1$ or find a violated
    inequality for $(\alpha,\beta,\theta)$ for Dual$_\mathcal{T}^b$ for
    $b=O(\log^2 k \log n)$.
\end{lemma}

Once again, we defer the proof to the appendix, but sketch the argument. To efficiently separate,
we observe that constraint~\ref{dineq1} can be recast as a covering Steiner tree problem. Using approximation
algorithms for this problem, we find a violated inequality for a (stronger) constraint. This results in the ``approximate
separation oracle''.

\begin{theorem}
There exists a $O(\log^2 k \log n)$ approximation to the Latency-SAND problem.
\end{theorem}

\begin{proof}
Combining Lemma 3.2 of~\cite{CS16} with Lemma~\ref{lem:apxseparate}, we can now compute an $\epsilon$-additive optimal solution to LP$_\mathcal{T}^b$ for $b = O(\log^2 k \log n)$. Using Lemma~\ref{lem:round}, we then achieve an $O(b)$ approximation for our problem.
\end{proof}

\bibliography{sand}

\begin{thebibliography}{10}

\bibitem{BCIW12}
M.~Balcan, F.~Constantin, S.~Iwata, and L.~Wang.
\newblock Learning valuation functions.
\newblock In {\em Conference on Learning Theory}, volume~23, pages 4--1, 2012.

\bibitem{BK05}
W.~Ben-Ameur and H.~Kerivin.
\newblock Routing of uncertain demands.
\newblock {\em Optimization and Engineering}, 3:283--313, 2005.

\bibitem{BR11}
K.~Bhawalkar and T.~Roughgarden.
\newblock Welfare guarantees for combinatorial auctions with item bidding.
\newblock In {\em Proceedings of the twenty-second annual ACM-SIAM symposium on
  Discrete Algorithms}, pages 700--709. Society for Industrial and Applied
  Mathematics, 2011.

\bibitem{CS16}
D.~Chakrabarty and C.~Swamy.
\newblock Facility location with client latencies: {LP}-based techniques for
  minimum-latency problems.
\newblock {\em Mathematics of Operations Research}, 41(3):865--883, 2016.

\bibitem{C07}
C.~Chekuri.
\newblock Routing and network design with robustness to changing or uncertain
  traffic demands.
\newblock {\em {SIGACT} News}, 38(3):106--128, 2007.

\bibitem{CC02}
M.~Chlebik and J.~Chlebikova.
\newblock Approximation hardness of the steiner tree problem on graphs.
\newblock {\em Proceedings of the Scandinavian Workshop on Algorithm Theory},
  pages 170--170, 2002.

\bibitem{Duff99}
N.G. Duffield, P.~Goyal, A.G. Greenberg, P.P. Mishra, K.K. Ramakrishnan, and
  J.E. van~der Merwe.
\newblock A flexible model for resource management in virtual private networks.
\newblock {\em Proceedings of {SIGCOMM}}, 29:95--108, 1999.

\bibitem{FRT04}
J.~Fakcharoenphol, S.~Rao, and K.~Talwar.
\newblock A tight bound on approximating arbitrary metrics by tree metrics.
\newblock {\em Journal of Computer and System Sciences}, 69:485--497, 2004.

\bibitem{F09}
U.~Feige.
\newblock On maximizing welfare when utility functions are subadditive.
\newblock {\em SIAM Journal on Computing}, 39(1):122--142, 2009.

\bibitem{FKPS16}
A.E. Feldmann, J.~K\"onemann, K.~Pashkovich, and L.~Sanit\`a.
\newblock Fast approximation algorithms for the generalized survivable network
  design problem.
\newblock {\em Proceedings of {ISAAC} (International symposium on algorithms
  and computation)}, pages 33:1-- 33:12, 2016.

\bibitem{Fing97}
J.~Fingerhut, S.~Suri, and J.~Turner.
\newblock Designing least-cost nonblocking broadband networks.
\newblock {\em Journal of Algorithms}, 24(2):287--309, 1997.

\bibitem{GW95}
M.X. Goemans and D.P. Williamson.
\newblock A general approximation technique for constrained forest problems.
\newblock {\em SIAM Journal on Computing}, 24(2):296--317, 1995.

\bibitem{GOS11}
N.~Goyal, N.~Olver, and F.~B. Shepherd.
\newblock Dynamic vs. oblivious routing in network design.
\newblock {\em Algorithmica}, 61(1):161--173, 2011.

\bibitem{GOS13}
N.~Goyal, N.~Olver, and F.~B. Shepherd.
\newblock The {VPN} conjecture is true.
\newblock {\em Journal of the ACM}, 60(3):17:1--17:17, June 2013.

\bibitem{GRS11}
F.~Grandoni, T.~Rothvo\ss, and L.~Sanit\`a.
\newblock From uncertainty to non-linearity: Solving virtual private network
  via single-sink buy-at-bulk.
\newblock {\em Mathematics of Operations Research}, 36(2):185--204, 2011.

\bibitem{GKKRY01}
A.~Gupta, J.~Kleingerg, R.~Kumar, B.~Rastogi, and B.~Yener.
\newblock Provisioning a virtual private network: A network design problem for
  multicommodity flow.
\newblock {\em Proceedings of Symposium on Theory of Computing (STOC)}, pages
  389--398, 2001.

\bibitem{GNR10}
A.~Gupta, V.~Nagarajan, and R.~Ravi.
\newblock An improved approximation algorithm for requirement cut.
\newblock {\em Operations Research Letters}, 38(4):322--325, 2010.

\bibitem{Gupta-Srinivasan}
A.~Gupta and A.~Srinivasan.
\newblock On the covering steiner problem.
\newblock In {\em International Conference on Foundations of Software
  Technology and Theoretical Computer Science}, pages 244--251. Springer, 2003.

\bibitem{J01}
K.~Jain.
\newblock A factor 2 approximation algorithm for the generalized steiner
  network problem.
\newblock {\em Combinatorica}, 21(1):39--60, 2001.

\bibitem{KR00}
G.~Konjevod and R.~Ravi.
\newblock An approximation algorithm for the covering steiner problem.
\newblock In {\em Proceedings of the Eleventh Annual ACM-SIAM Symposium on
  Discrete Algorithms}, SODA '00, pages 338--344, 2000.

\bibitem{LLN01}
B.~Lehmann, D.~Lehmann, and N.~Nisan.
\newblock Combinatorial auctions with decreasing marginal utilities.
\newblock In {\em Proceedings of the 3rd ACM conference on Electronic
  Commerce}, pages 18--28. ACM, 2001.

\bibitem{OSZ13}
G.~Oriolo, L.~Sanit{\`a}, and R.~Zenklusen.
\newblock Network design with a discrete set of traffic matrices.
\newblock {\em Operations Research Letters}, 41(4):390--396, 2013.

\bibitem{Y78}
M.~Yannakakis.
\newblock Node-and edge-deletion np-complete problems.
\newblock In {\em Proceedings of the tenth annual ACM symposium on Theory of
  computing}, pages 253--264. ACM, 1978.

\end{thebibliography}

\appendix
\section{Obstacles to an $O(1)$-approximation}
\label{sec:obstacles}

\sand{}, while being a generalization of well-studied problems, cannot be easily
approximated by many common approximation techniques. This section
presents some of the obstacles to getting a $O(1)$-approximation for \sand{}.

\subsection{Tree solutions are far from optimal}
One of the most interesting (and frustrating) aspect of this problem is that
a tree solution can be far from optimal. This
rules out many algorithms used for other network design problems (such as the
primal-dual ``moat-growing") as the resulting
solution produced is a tree.

\begin{lemma}
    There exists a family of instances with a $\Omega(\log n)$ gap between the best
    solution on a graph and the best solution which is a tree.  We emphasize that
    $n$ is not the length of the input but the number of nodes in the graph.
\end{lemma}

\begin{proof}
    We adapt a construction from~\cite{GOS11}.  Consider an expander graph
    $G=(V,E)$ with constant degree $d \geq 3$ and edge-expansion at least 1. Now
    let $G'$ be $G$ with an additional node $r$, where $r$ has an edge to every
    vertex in $V$. Let the colors consist of all sets of size $b=\log n$ of the
    vertices, $\binom{V}{b}$.  All edges adjacent to the root $r$ have cost $\frac{n}{b}$
    and the remaining edges have cost $1$.

    One valid graph solution on $G'$ is to take any $b$ edges adjacent to $r$ and
    then all the edges in $G$. Consider any color $C_i$. Consider any subset $S$ of
    vertices of $V$. By the expansion property of $G$, the number of edges
    adjacent to $S$ is at least $|S|$ when $|S| \leq n/2$, so $\delta(S)
    \geq |S\cap C_i|$. When $|S| >n/2$, at least $n-|S|$ edges from $G$ and at least $b-(n-|S|)$ edges adjacent
    to $r$. Hence, $\delta(S) \geq b$ when $|S| > n/2$ implying $\delta(S) \geq |S\cap C_i|$.
    Therefore a flow from $C_i$ to $r$ exists and this is a feasible solution. This
    solution has cost $3 \cdot n+b \cdot n/b=4n = O(n)$.

    The original paper~\cite{GOS11} presenting this example shows that any tree solution
    has cost $\Omega(n\log n)$. They prove this by arguing that either a large number
		of $b$ cost edges adjacent must be used or there must be a large number
		of long paths which are disjoint.
\end{proof}
		
    Therefore, if any algorithm always built a tree, we can't expect to do better
    than $\Omega(\log n)$. We can achieve $O(\log n)$ with a tree solution by
    simply using FRT though. Therefore, any algorithm which improves on $O(\log n)$
    must avoid always build a tree.

\subsection{A Bad Case for Iterative Rounding}
We now give a set of examples which seems to rule out the iterative rounding approach
due to Jain~\cite{J01}.
This example comes from a special case of Kneser Graphs known as Odd Graphs.
Odd Graphs, denoted $O_s$, contain a vertex for every $s$-size
subset of $[2s+1]$. We connect to sets $S,T \subseteq[2s+1]$ in this graph
iff $S \cap T = \emptyset$. The graph has $n=\binom{2s+1}{s}$ vertices and
every vertex has degree $s+1$.

Our instance considers $O_s$ with an additional node $r$ which will serve as the root. We
connect the root $r$ to all other nodes in $O_s$.  The edges adjacent to the
root $r$ all have weight $2$ and all the remaining edges have weight $1$. In a
For any edge $(u,v) \in E$, we define a color $C_{uv}=(\{u\}\cup N(u)) -\{v\}$ (i.e.~$u$
and $u$'s whole neighborhood except for $v$).

One valid fractional solution is for all the edges $e$, adjacent to the root to
have $x_e=\frac{s+1}{n}$, and all the edges $e$, not adjacent to the root to
have weight $x_e = \frac{s+1}{s^2+1}-\frac{(s+1)^2}{(s^2+1)n}$. So, the total
cost is:
\[ 2n\frac{s+1}{n} + \frac{n(s+1)}{2}(\frac{s+1}{s^2+1}-\frac{(s+1)^2}{(s^2+1)n})
    \approx 2(s+1)+\frac{n}{2}-\frac{s+1}{2}\approx \frac{n}{2}+\frac{3}{2}(s+1) \]

Any valid integral solution must have every node attached to the root.
Therefore, the cost of the minimum spanning tree is a lower bound on the cost
of the best integral solution. The cost of a minimum spanning tree is at least
$n+1$; there are $n+1$ nodes in the whole graph, and all the edges have weight
1 except for all the edges adjacent to the root which all have weight 2.

As $s$ increases, then we get the ratio between the best fractional solution
and the best integral solution goes to $2$. Therefore, the best approximation
we can hope for when using this LP is $2$.

In addition, we believe that the solution we gave above is an extreme point
and can numerically verify this for $s \leq 15$ (a graph containing over $300$ million vertices).
The values $x_e$ get arbitrarily close to 0 as $s$ increases. This indicates
that iterative rounding on this LP will not yield a constant approximation.

\section{Missing Proofs in Section 3}

\subsection{Proof of Claim~\ref{clm:nodouble}}

\begin{proof}
In order to prove this claim, it is enough to show that if $\mathcal S$ does not satisfy the condition of the claim for some subset of terminals $\mathcal C \subseteq C_1 \cup C_2$, then we can construct another solution $\mathcal S'$, whose cost is at most the cost of $\mathcal S$, that does not satisfy the condition of the claim for some subset of nodes $\mathcal C' \subset \mathcal C$.

Let $z \in C_q$ be an arbitrary terminal for which the condition of the claim is not satisfied, with $q \in \{1,2\}$,
and let $y$ be the first node in $\bigcup_{i=1}^p\{y_i^1,y_i^3,y_i^5,y_i^7\}$ on its path $P$ to $r$. By assumption, $P$
does not contain the edge $e:=\{y,r\}$.

Suppose first that $ y=y_i^1$ for some index $i$ (the proof for $y=y_i^7$ is similar). Then, necessarily, $z=y_i^1$, and the first two edges of $P$ are $\{y_i^1, y_i^2\}$ and $\{y_i^2, y_i^3\}$. We construct $\mathcal S'$ by changing $P$ with the path formed by the single edge $e=\{y_i^1,r\}$. If $e$ is in the support of $\mathcal S$, by Fact 2, there can be only terminals in $C_2$ using it, and therefore in $\mathcal S'$ we do not need to increase the capacity of any edge.
If instead the edge $e$ is not in the support of $\mathcal S$, then the node $y_i^1$ is a node of degree one in the support of $\mathcal S$. When we change the path for $y_i^1$, in $\mathcal S'$ we have to increase the capacity of $e$ to one. On the other hand, we can decrease the amount of capacity of the edge $\{y_i^1, y_i^2\}$ by one, since no other terminal is using that edge. As for the edge $\{y_i^2, y_i^3\}$, note that this edge is either not used by any terminal, or it is used by one terminal $k_j \in C_1 \cap C_2$ for some $j$. In both cases, we can decrease its amount of capacity by one. It follows that the cost of $\mathcal S'$ is at most the cost of $\mathcal S$, and $\mathcal C' \subset \mathcal C$.

Suppose now that $y=y_i^3$ for some index $i$ (the proof for $y_i^5$ is similar).  Then, necessarily, the first two edges of the subpath $\tilde P \subseteq P$ from $y_i^3$ to $r$, are either $\{y_i^3, y_i^2\}$ and $\{y_i^2, y_i^1\}$, or $\{y_i^3, y_i^4\}$ and $\{y_i^4, y_i^5\}$, and we have therefore two cases.

Case A: the first two edges of $\tilde P$ are $\{y_i^3, y_i^2\}$ and $\{y_i^2, y_i^1\}$. Let $\bar q \in \{1,2\}$ be different from  $q$. If there is no terminal in $C_{\bar q}$ whose path to $r$ uses these edges, then changing  $\tilde P$ to $e=\{y_i^3,r\}$ implies increasing the capacity on $e$ by one, and decreasing the capacity of the edges $\{y_i^3, y_i^2\}$ and $\{y_i^2, y_i^1\}$ by one. This yields a solution $\mathcal S'$ with $\mathcal C' \subset \mathcal C$, whose cost is no greater than the cost of $\mathcal S$. Suppose now that there is a terminal $\bar z$ in $C_{\bar q}$ whose path uses both the edges $\{y_i^3, y_i^2\}$ and $\{y_i^2, y_i^1\}$. Note that $\bar z \neq y_i^1$, since otherwise $y_i^1 \in \mathcal C$ but by the previous argument we can assume that this is not the case. Therefore, necessarily, the flow going from $\bar z$ to $r$ travels first on the edge $\{y_i^3, y_i^2\}$ and then on
$\{y_i^2, y_i^1\}$ (i.e. the flow from $z$ to $r$ and the flow from $\bar z$ to $r$ induce the same orientation on these edges). Then, in the solution $\mathcal S'$ we change the paths for both $z$ and $\bar z$, by substituting the subpaths from $y_i^3$ to $r$ with the edge $e$. Once again, this implies increasing the capacity on $e$ by one, and decreasing the capacity of the edges $\{y_i^3, y_i^2\}$ and $\{y_i^2, y_i^1\}$ by one, and therefore $\mathcal S'$ has $\mathcal C' \subset \mathcal C$, and cost no greater than the cost of $\mathcal S$.
Finally, suppose that there is no terminal in $C_{\bar q}$ whose path uses both the edges $\{y_i^3, y_i^2\}$ and $\{y_i^2, y_i^1\}$, but there is a terminal $\bar z \in C_{\bar q}$ that uses exactly one of them. Then, $\bar z = k_{i_1}$. If $P_{\bar q}^{i_1}$ is the path with nodes $\{k_{i_1}, y_i^2, y_i^1, r\}$, then we construct $\mathcal S'$ by changing $P_{\bar q}^{i_1}$ to the path with nodes $\{k_{i_1}, y_i^2, y_i^3, r\}$, and by changing $\tilde P$ to $e$. This implies decreasing the capacity on the edges $\{y_i^2, y_i^1\}$
and $\{y_i^1, r\}$ by one, and increasing by one the capacity on the edges $\{y_i^2, y_i^3\}$ and $\{y_i^3, r\}$. One can see that $\mathcal S'$ has  $\mathcal C' \subset \mathcal C$, and cost no greater than the cost of $\mathcal S$. If instead
$P_{\bar q}^{i_1}$ is not the path with nodes $\{k_{i_1}, y_i^2, y_i^1, r\}$, then necessarily the second edge of $P_{\bar q}^{i_1}$ is $\{y_i^2, y_i^3\}$. By Fact 2, this implies that there is no terminal of $C_{\bar q}$ using the edge $\{y_i^1, r\}$, other than possibly $y_i^1$.
Furthermore, by Fact 1, we know that also the first edge of $P_q^{i_1}$ is $\{k_j, y_i^2\}$, and by Fact 2, the other edges of $P_{q}^{i_1}$ are $\{y_i^2, y_i^1\}$ and $\{ y_i^1,r\}$. This implies that the capacity of the edge $\{y_i^1, r\}$ is at least 2 in $\mathcal S$. If we construct $\mathcal S'$ by changing $\tilde P$ to $e$, we can decrease the capacity of the edge $\{y_i^1, r\}$ by one, and increase the capacity on $e$ by one. Once again, the result follows.

Case B: the first two edges of $\tilde P$ are $\{y_i^3, y_i^4\}$ and $\{y_i^4, y_i^5\}$. If there is no terminal in $C_{\bar q}$ using the edges $\{y_i^3, y_i^4\}$ and $\{y_i^4, y_i^5\}$, or if there is a terminal $\bar z \in C_{\bar q}$ that uses (at least one of) these edges, but such that the flow going from $z$ to $r$ and the flow from $\bar z$ to $r$ induce the same orientation on these edges, then the statement follows similarly to the previous Case A. The only additional case we have to handle here, arises if there is a terminal $\bar z \in C_{\bar q}$ that uses the edge $\{y_i^3, y_i^4\}$ but such that the flow from $\bar z$ to $r$ and the flow from $z$ to $r$ induce opposite orientation on this edge. In this case, by the arguments of Case A, we can assume that the flow from $\bar z$ to $r$ travels on the edge $e=\{y_i^3, r\}$. Therefore, if there is no terminal in $C_q$ using $e$ we construct $\mathcal S'$ by simply changing $\tilde P$ to $e$. One sees that $\mathcal S'$ has the same cost as $\mathcal S$ and $\mathcal C' \subset \mathcal C$.
Suppose instead that there is some terminal in $C_q$ using $e$. Then, necessarily $q\neq 1$ (because the only terminal in $C_1$ different from $z$ that could potentially use $e$ without violating Fact 2 is $y_i^1$, but we already argued that $y_i^1$ routes on the edge $\{y_i^1,r\}$). Furthermore, using again Fact 2, we can say that there is exactly one terminal $\tilde z$ in $C_2$ routing on $e$, and we have either $\tilde z = y_i^3$ and $\bar z = k_{i_1}$, or $\bar z = y_i^3$ and $\tilde z = k_{i_1}$. In either case, by Fact 1 and Fact 2, we know that $P_1^{i_1}$ must contain the edges $\{y_i^2, y_i^1\}$ and $\{y_i^1, r\}$. We construct $\mathcal S'$ by changing $\tilde P$ to $e$, and by changing $P_2^{i_1}$ to be equal to $P_1^{i_1}$. One can see that these changes do not require increasing the capacity on any edge, and therefore the result follows.
\end{proof}

\subsection{Proof of Claim~\ref{clm:second}}
\begin{proof}
First, observe that a straightforward corollary of Claim~\ref{clm:nodouble} is that every terminal $y \in \bigcup_{i=1}^p\{y_i^1,y_i^3,y_i^5,y_i^7\}$ sends flow to $r$ on the path $\{y,r\}$. Let $\tilde n_i \leq n_i$ be the number of nodes $k_j$ that use edges of $G_i$ to route flow to $r$ and, in addition, satisfy $P_1^j = P_2^j$. Using the previous observation, the cost of the capacity on the edges $E_1 \cap E(G_i)$ is at least $2(4+\tilde n_i)$. Furthermore, another consequence of Claim 1 is that if a terminal $k_j$ uses edges of $G_i$ to route flow to $r$, and the paths $P_1^j$ and $P_2^j$ share the second edge, then $P_1^j=P_2^j$ and these paths contain exactly one edge of $E_2$ (and one edge of $E_1$). It follows that the capacity installed on the edges $E_2 \cap E(G_i)$ is exactly $\tilde n_i + 2(n_i -\tilde n_i) = 2n_i -\tilde n_i$. Putting things together,
\begin{equation}\label{eq:cost}
\chi_i \geq 2(4+\tilde n_i) +2n_i - \tilde n_i = 8 + 2n_i + \tilde n_i \geq 8 + 2n_i
\end{equation}
We now prove that if $x_i$ is in conflict, then either the first or the last inequality of (\ref{eq:cost}) is strict. If $x_i$ is in conflict, then $n_i >1$.
Suppose $n_i =2$, and let $k_j, k_{j'}$ be the terminals that use edges of $G_i$ to route flow to $r$.
Without loss of generality, let the endpoint of the first edge of $P_j^1$ be $y_i^2$, and
the endpoint of the first edge of $P_{j'}^1$ be $y_i^4$ (the other case is similar). If $\tilde n_i>0$ then the last inequality of (\ref{eq:cost}) is strict. If instead $\tilde n_i=0$, we know that $P_1^j$ and $P_2^j$ as well as $P_1^{j'}$ and $P_2^{j'}$
do not share the second edge. In this case, the cost of the capacity installed on the edges $E_2 \cap E(G_i)$ is 4. However, we need 2 units of capacity on at least one edge of $E_1 \cap E(G_i)$, implying that the cost of the capacity installed on the edges $E_1 \cap E(G_i)$ is at least 10, and therefore $\chi_i$ is at least $14 > 8+2n_i$.

Suppose $n_i =3$. As in the previous case, if $\tilde n_i >0$ the last inequality of (\ref{eq:cost}) is strict. If instead $\tilde n_i=0$, then all the paths used by these $n_i$ terminals to route flow to $r$ do not share the second edge. Then, the cost of the capacity installed on the edges $E_2 \cap E(G_i)$ is 6. However, in this case we also need 2 units of capacity on at least two edges of $E_1 \cap E(G_i)$, implying that the cost of the capacity installed on the edges $E_1 \cap E(G_i)$ is at least 12, and therefore $\chi_i \geq 18 > 8+2n_i$.
\end{proof}

\section{Missing Proofs in Section 4}

\subsection{Proof of Lemma~\ref{lem:round}}
\begin{proof}
For each time $t=2^i$ for $i=0, \dots, \log L$, sample a tree based on the distribution $z_{\mathcal{T}, t}$ - here $L$ is an upper bound on the length of a tree that covers all colors, such as the value of an MST on them. Denote the tree sampled at time $t$ by $T_t$. We ``Eulerify" (walk around the tree) $T_t$ into a path $P_t$ which starts and ends at $r$. This path $P_t$ has length at most $2bt$. Now our solution will be the path formed by concatenating $P_1, P_2, P_4, \dots, P_{L}$. This is a feasible solution because at time $2L$ we are able to pick up all the nodes of all colors by taking a path around a spanning tree. (By using standard scaling methods~\cite{CS16}, we can reduce the $\log L$ factor in this sampling to $O(\log n)$ to get a strongly polynomial algorithm - we omit the details.)

Now we show that the expected cost of this solution is $O(b)$ times the cost of the linear program. Let $t_j$ denote the first time when $\sum_{t} x_{j,t} \geq 2/3$. Then, we know the contribution of the $x_{j,t}$ to the objective function is at least $\frac{1}{3}t_j$. For every time $t' \geq t_j$ the probability the tree we pick has $j$ elements of each color is at least $2/3$ by our definition of $t_j$. So, the expected length of the path before we get $j$ elements of each colors in the path is given by a geometric sum.
\begin{align*}
\sum_{i=\log t_j}^\infty \Pr[ & \text{ $j\in T_{2^i}$ and $j\notin T_{2^{i-1}}$}] \text{(Cost of the first $i$ trees)}  \\&=
\sum_{i=\log t_j}^\infty \frac{2}{3}\left(\frac{1}{3}\right)^{i-\log t_j} (4b2^i)\\
&= \frac{8}{3} b2^{\log t_j} \sum_{i=\log t_j}^\infty 2^{i-\log t_j} \left(\frac{1}{3}\right)^{i-\log t_j}\\
&= \frac{8}{3} bt_j \sum_{i=0}^\infty  \frac{2}{3}^{i}\\
&= \frac{16}{3} bt_j
\end{align*}
Since the contribution of the $x_{j,t}$ to the objective function is at least $\frac{1}{3}t_j$, the expected cost of the whole path is at most $O(b)$ times the cost of the LP.
\end{proof}

\subsection{Proof of Lemma~\ref{lem:apxseparate}}

\begin{proof}
Given $(\alpha,\beta,\theta)$,
it is easy to verify all constraints except for Equation~\ref{dineq1}.
We interpret $\theta_{j,t}$ to be the reward for collecting $j$ terminals of each color by time $t$ and $\beta_t$ to be the budget by time $t$. This constraint says that we want to ensure that every path (tree in the relaxation) rooted at
$r$ of length at most $bt$ can achieve no more than $\beta_t$ rewards.

To find a violated constraint, we solve a covering Steiner tree problem~\cite{KR00}, one for each $j$: the problem is defined by $k$ groups one for each color $C_i$ and we require that each group has $j$ terminals connected using the smallest length tree. Using the best-known approximation algorithm for covering Steiner tree~\cite{Gupta-Srinivasan} we obtain an approximation factor of $b=O(\log^2 k \log  n)$.  Now we can check for each time $t$, if the rewards collected are bounded above by the budget $\beta_t$  where we allow trees of size $b t$.
\end{proof}

\end{document}